\documentclass[10pt,conference,letterpaper]{IEEEtran}

\usepackage{tikz} 
\usetikzlibrary{arrows}

\usepackage{float}
\usepackage{graphicx}
\usepackage[caption=false,font=footnotesize]{subfig}
\usepackage{makecell}
\usepackage{amsmath}
\usepackage{amssymb}
\usepackage{hyperref}

\usepackage{amsthm}

\theoremstyle{plain}
\newtheorem{property}{Property}
\newtheorem{lemma}{Lemma}
\newtheorem{definition}{Definition}
\newtheorem{theorem}{Theorem}

\usepackage{threeparttable}

\usepackage{xspace}
\newcommand{\MBDAG}{$(T,t_0$,$S_B)$-Multithreaded block DAG\xspace}

\begin{document}

\title{Blockclique: Scaling Blockchains \\ through Transaction Sharding \\ in a Multithreaded Block Graph}

\author{\IEEEauthorblockN{S\'ebastien Forestier, Damir Vodenicarevic, Adrien Laversanne-Finot}
\IEEEauthorblockA{Massa Labs, Paris, France \\ contact@massa.network}
}

\maketitle

\begin{abstract}
Decentralized crypto-currencies based on the blockchain architecture under-utilize available network bandwidth, making them unable to scale to thousands of transactions per second. We define the Blockclique architecture, that addresses this limitation by sharding transactions in a block graph with a fixed number of threads. The architecture allows the creation of intrinsically compatible blocks in parallel, where each block references one previous block of each thread. The consistency of the Blockclique protocol is formally established in presence of attackers. An experimental evaluation of the architecture's performance in large realistic networks demonstrates an efficient use of available bandwidth and a throughput of thousands of transactions per second.
\end{abstract}

\section{Introduction}

In a decentralized crypto-currency network, any node can join the network, fetch from peers and verify the history of executed transactions, create and broadcast new candidate transactions, and execute sets of candidate transactions. 

In order to achieve a global consensus on the history of transactions, the protocol must regulate the execution of transactions.
Decentralized protocols perform a Sybil-resistant selection of nodes that are allowed to execute transactions in a timely manner.
Proof-of-Work, used in Bitcoin~\cite{nakamoto2008bitcoin} and other protocols, selects random nodes to create blocks of transactions depending on the nodes' use of computational power.
Proof-of-Stake (e.g. in Tezos~\cite{goodman2014tezos}) randomly selects block producers based on the amount of coins they hold.

Even with a regulated selection of block producers, the latency in peer-to-peer networks can cause different nodes to produce or observe incompatible blocks, requiring protocol-defined consensus rules to specify which blocks of transactions should be preferred.
In Bitcoin, the Nakamoto rule states that the chain with the most work should be preferred and built upon.
However, if too many blocks are produced and the network can't broadcast them fast enough, the rate of incompatible (stale) blocks can become too high and the consensus can fail.

Blockchain protocols keep the stale rate low by limiting block size and frequency, which also reduces network usage and limits transaction throughput ($5$ tx/s in Bitcoin).
The median bandwidth of Bitcoin nodes is $56$\,Mb/s~\cite{gencer2018decentralization} but a single $8$\,Mb block is propagated every $600$\,s on average, leaving network connections largely underused.
As a result, the consensus bitrate $C_B$ defined as the average bitrate of blocks assuming they are full, is $C_B=13$\,kb/s in Bitcoin.
The consensus bitrate is set by protocol parameters, but effective network properties such as the number of nodes and their bandwidth and latency make it possible or not to sustain that bitrate with a low stale rate.
Although Bitcoin's consensus bitrate could be safely increased by some margin, as illustrated by the Bitcoin Cash increased block size, the thousand-fold difference between the consensus bitrate and the actual bandwidth of peers underlines the inefficiency of relying on a single chain~\cite{croman2016scaling}.

This paper introduces the Blockclique architecture, a natural extension of blockchains which optimizes network usage by parallelizing block creation into $T$ threads.
Blockclique uses transaction sharding to ensure that the transactions contained in blocks created in parallel are always compatible: a block $b$ in a thread $\tau$ can only include transactions with input addresses assigned to the thread $\tau$, while transaction outputs can belong to any thread.
Simply using $T$ separate blockchains would however split network resources and degrade the security of the protocol by a factor $T$ compared to a single blockchain. 
Therefore, we do not shard the network of nodes, so that all nodes produce and verify blocks in all threads.
Different threads are synchronized by linking their blocks together in a directed acyclic graph structure (DAG) where each block acknowledges one parent block in each thread by including its hash.
Contrary to network sharding solutions, all nodes process all blocks of all threads so that there is no need for cross-shard communication other than the cross-thread parent links in the graph.

\begin{table*}[t]
\centering
\normalsize
\begin{threeparttable}[b]
\caption{Comparison of Decentralized Crypto-Currency Protocols.}
\begin{tabular}{|c||c|c|c|c|c|c|}
 \hline
 Protocol & \begin{tabular}[c]{@{}c@{}}Data\\ Structure\end{tabular} & Sharding & \begin{tabular}[c]{@{}c@{}}Sybil\\ Resistance\end{tabular} & \begin{tabular}[c]{@{}c@{}}Consensus\\ Family\end{tabular} & Security\,\tnote{a} & Throughput\,\tnote{b} \\
 \hline
 Bitcoin~\cite{nakamoto2008bitcoin} & Block tree & No & Proof-of-Work & Nakamoto & 50\%  & 7 tx/s\\
 SPECTRE~\cite{sompolinsky2016spectre} & Block DAG & No & Proof-of-Work & Nakamoto & 50\% & Not Avail.\,\tnote{c} \\
 Conflux~\cite{li2018scaling} & Block DAG & No & Proof-of-Work & Nakamoto & 50\% & 6400 tx/s\,\tnote{c} \\
 OHIE~\cite{yu2018ohie} & Parallel Trees & No & Proof-of-Work & Nakamoto & 50\% & 2420 tx/s\,\tnote{c}   \\
 Avalanche \cite{rocket2019scalable} & Tx DAG & No & Any & Metastability & 20\%  & 3400 tx/s\\
 Elastico~\cite{luu2016secure} & Block tree & Network+Tx & Proof-of-Work & Byzantine & 25\%& 16 bx / 110s \\
 Omniledger~\cite{kokoris2018omniledger} & Block DAG & Network+UTXO & Any & Byzantine & 25\%  & 500 tx/s\\
 \textbf{Blockclique (ours)} & Multithread. DAG & Transaction & Any & Nakamoto & 45\%  & 10000 tx/s \\
 \hline
\end{tabular}
\begin{tablenotes}
\item [a] \small Maximum resource proportion of attackers under which the protocol is secure. Threat models may differ.
\item [b] Assumptions on the number of nodes in the network and their bandwidth may differ.
\item [c] Non-unique transactions. Transactions can appear multiple times in the structure, reducing the effective throughput.
\end{tablenotes}
\end{threeparttable}
\label{table:comparison}
\end{table*}

After defining the blockclique data structure, we establish a Blockclique-specific consensus mechanism derived from the Nakamoto consensus rule.
We then study the security of the protocol, formally prove its consistency, and derive optimal parameters for security and performance.
In our network simulations, the Blockclique architecture exceeds $10,000$ transactions per second with a transaction time of less than a minute, in a large network with realistic properties.
Overall, our results show that it is not necessary to adopt radically different blockchain protocols to obtain a scaled and decentralized currency, and that a natural parallelization of blockchains makes an efficient use of peer-to-peer networks.

\section{Related Work}

Previous attempts at scaling decentralized blockchains through sharding and/or changes in data structure are especially relevant to our work, and are summarized in Table~\ref{table:comparison}.

\subsubsection{Changes in Data Structure}

One line of work seeks to scale blockchains by extending the classical block tree structure to a structure allowing a parallel production of blocks and transactions. 
The first directed acyclic block graph (block DAG) structures appear in \cite{inclusive,sompolinsky2015secure,sompolinsky2016spectre}.
In SPECTRE~\cite{sompolinsky2016spectre}, nodes create blocks in parallel in a block DAG, and a voting process sorts transactions and chooses which ones are executed.
Similarly, the Conflux \cite{li2018scaling} and OHIE \cite{yu2018ohie} protocols allow the creation of parallel blocks in a DAG or a set of parallel chains.
However, as those different protocols do not implement transaction sharding, the parallel blocks can contain the same transactions many times which can drastically reduce the effective transaction throughput.

In IOTA~\cite{tangle}, transactions are included in a transaction DAG. 
To emit a new transaction, a node attaches it to two tip transactions of its local DAG, solves a small Proof-of-Work puzzle and broadcasts the transaction.
A coordinator run by the IOTA foundation provides checkpoints every minute so that nodes consistently verify transactions, and so that the DAG does not grow excessively in width. 
Avalanche~\cite{rocket2019scalable} uses a data structure similar to IOTA. 
Users are free to choose which transactions they want to reference, and therefore need to be incentivized to help build a DAG with limited width.
Blockclique also uses a DAG structure, but restricts it to a fixed number of threads, allowing the DAG to grow only in one direction. 
This greatly simplifies the protocol and its analysis, and removes the need for central entities \cite{tangle} or incentives for users to grow the DAG in a single direction~\cite{rocket2019scalable}.

\subsubsection{Sharding}

Sharding consists in distributing nodes and/or transactions into several groups (``shards'') for parallel processing. 
Most existing sharding protocols rely on network sharding: nodes are divided into groups, each processing a given subset of the data \cite{luu2016secure,kokoris2018omniledger,zilliqa}.
A ``directory'' group is then responsible for aggregating blocks coming from all shards into a single blockchain.
Group members are typically selected using a PoW puzzle. 
For resilience against attackers with a large fraction of the computational power, each group must contain a large number of members. 
However, consensus within a group is typically achieved using classical Byzantine Fault Tolerant protocols, which do not scale well~\cite{bano2017consensus,berger2018scaling} and cause the transaction throughput to decrease with group size.
As a result, such schemes face a security-performance dilemma.

Blockclique shards transactions in order to parallelize block creation, and does not rely on network sharding. 
As a result, it is closer to traditional blockchain protocols: each participant is randomly selected to create blocks in all threads, and verifies blocks of all threads. 
A consensus rule applied by all nodes determines in all threads which blocks should be considered confirmed.

To our knowledge, Blockclique is the first protocol to combine a parallel block structure to improve transaction throughput with transaction sharding to avoid the duplication of transactions.

Lastly, it is possible to improve the efficiency of blockchain applications using off-chain peer-to-peer payment channels~\cite{decker2015fast}, in which payment promises guaranteed by on-chain deposits are processed between pairs of nodes. Off-chain promises are only settled on the blockchain periodically, or in case of fraud, which allows high transaction throughput and reduced transaction fees. However such networks are still experimental, do not handle large transactions and suffer from payment hub centralization~\cite{seres2019topological}. This paper does not focus on off-chain overlays, but Blockclique could be used as a high-throughput basis for off-chain payment, offering fast settlement and quick payment channel reconfiguration.

\section{Architecture}
\label{archi}

The Blockclique architecture is a combination of a data structure for the ledger, block and transaction structures, a Sybil-resistant selection of nodes, an incentive model and a consensus rule.
Those elements are described in the following sections.

\subsection{Data Structure}

\begin{figure*}[t]
\centering
\subfloat[Multithreaded block DAG]{\includegraphics[width=0.42\textwidth]{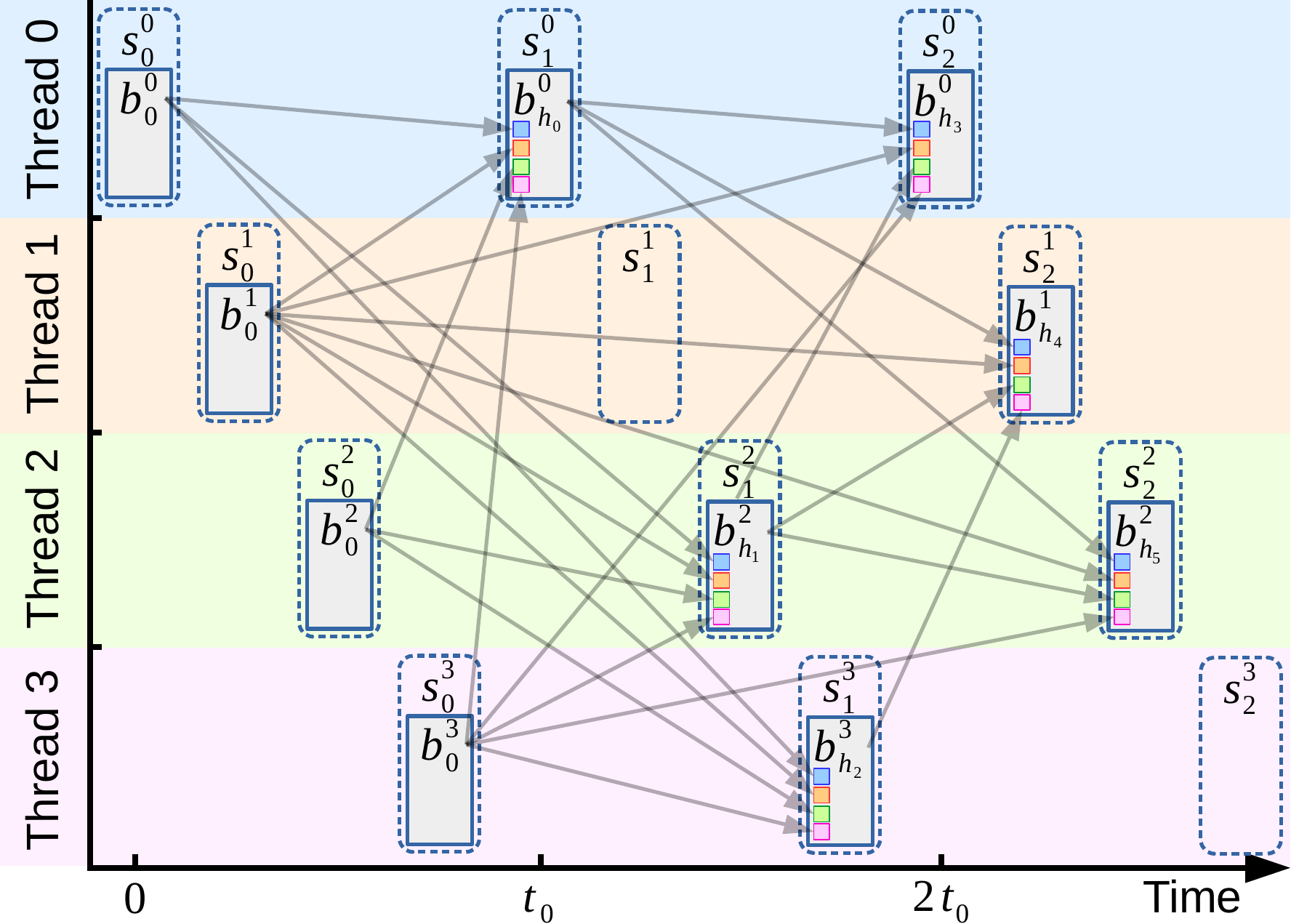}}\hspace{1.5cm}
\subfloat[Non-genesis block structure]{\raisebox{0.42cm}{\quad\includegraphics[width=0.2\textwidth]{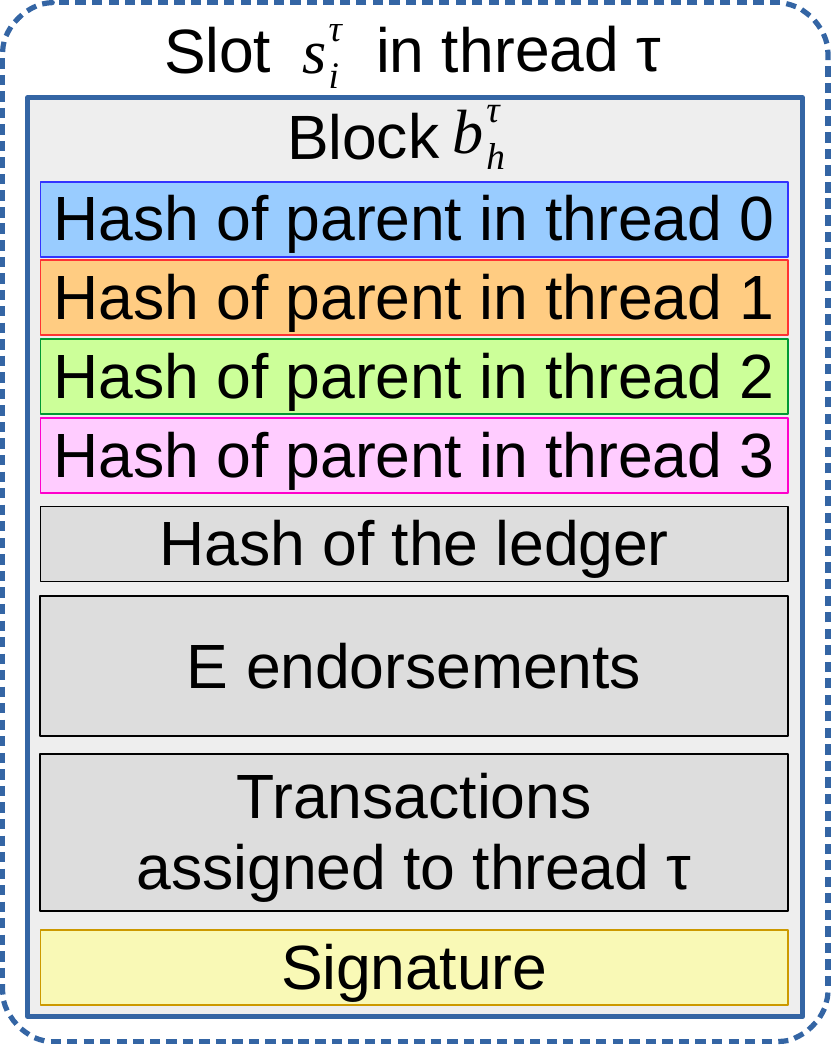}\quad}}
\caption{Data structure. (a) Example timeline of a multithreaded block DAG, with $T=4$ threads and one block slot every $t_0$ seconds in each thread. Dark arrows link parent blocks to their children.
The producer of block $b_{h_5}^2$ did not observe block $b_{h_2}^3$ yet but was still able to create a compatible block that references the earlier $b_{0}^3$ block instead. No block was broadcast for block slot $s^1_1$: it is a case of block miss.
(b) General structure of a non-genesis block $b_h^\tau$ in the block slot $s_i^\tau$ of thread $\tau$.}
\label{example1}
\end{figure*}

\subsubsection{Threads and block slots}

We define $T$ threads numbered from $\tau=0$ to $\tau=T-1$, each containing consecutive regularly spaced (by $t_0$ seconds) time slots that can host blocks. Fig.~\ref{example1}(a) shows an example timeline with $T=4$ threads.
The $i$-th block slot ($i\in \mathbb{Z}^+$) in thread $\tau$ is denoted by $s^\tau_i$ and occurs at $i\cdot t_0 + \tau t_0/T$ seconds.
The $\tau t_0/T$ time shift ensures that globally all block slots, and therefore network usage, are uniformly spread in time.

\subsubsection{Multithreaded Block DAG structure}
\label{multithreaded}

Blocks are identified by their cryptographic hash: $b_{h}^\tau$ refers to the block with hash $h$ in thread $\tau$.
We define one genesis block with no parents, denoted  $b_0^\tau$, in the first block slot of each thread.
Each non-genesis block references the hashes of $T$ parent blocks, one from each thread.
We define the parent function $P(b_{h_2}^{\tau_2}, \tau_1)$ that returns the parent in thread $\tau_1$ of a non-genesis block $b_{h_2}^{\tau_2}$.
The parent function generates a block graph $G$ where an arrow from $b_{h_1}^{\tau_1}$ to $b_{h_2}^{\tau_2}$ means $P(b_{h_2}^{\tau_2}, \tau_1) = b_{h_1}^{\tau_1}$.
As a child block includes its parent's hashes, and the hashing procedure of the child block takes into account those hashes, it is impossible in practice to build a cycle in the graph $G$, unless the security of the hashing function is compromised.
$G$ is therefore a directed acyclic graph of parallel blocks (block DAG).

\begin{definition}
Let \emph{\MBDAG} denote a block DAG structure with all the following properties:
\begin{itemize}
\item one genesis block is present in each of $T$ threads,
\item non-genesis blocks in thread $\tau$ reference one block of each thread as parents, have a size lower than $S_B$ bits, and a block slot number strictly higher than their parent's in thread $\tau$,
\item to ensure the consistency of block references, any ancestor $b_{h_1}^{\tau_1}$ of a block $b_{h_2}^{\tau_2}$ must be $P(b_{h_2}^{\tau_2}, \tau_1)$ or one of its ancestors. 
\end{itemize}
\end{definition}

Blockclique uses an \MBDAG as its data structure (see Fig.~\ref{example1}(a) for an example).

\subsubsection{Ledger}

In a high throughput architecture, the transaction history grows quickly, requiring nodes to be able to forget old blocks to save space.
In Blockclique, nodes store the balance of each address in a local ledger, so that they can verify whether the sender address of a transaction has enough coins, without looking up old transactions.

\begin{definition} 
Let $\mathcal{A}$ be the set of addresses and $\mathcal{B}$ the set of possible balances of an address.
Given a multithreaded block DAG $G$, a ledger state $\mathcal{L}(G,b_{h}^{\tau})$ is a mapping from addresses to their balances after the processing of block $b_{h}^{\tau}$ and its ancestors in $G$: $\mathcal{L}(G,b_{h}^{\tau}): \mathcal{A}\rightarrow \mathcal{B}$.
\end{definition}

\subsubsection{Blocks and Transactions}

The typical structure of a non-genesis block is shown in Fig.~\ref{example1}(b).

Blocks can contain transactions emitted by any node, up to a total block size limit of $S_B$ bits. A transaction represents a modification of the crypto-currency's ledger state, moving coins from one address to another.

Nodes are randomly selected to create blocks in particular block slots.
Furthermore, before each block slot $s^\tau_i$, $E$ randomly selected nodes are allowed to broadcast signed endorsements \cite{tezosdoc} carrying the hash of the last block in thread $\tau$, and those endorsements can be included in any of the $E$ endorsement slots within the block at slot $s^\tau_i$ by its creator.

\subsubsection{Transaction Sharding}

In the Blockclique protocol, transactions are sharded: they are deterministically divided into groups to be processed in parallel threads.
For instance, if there are $T=32$ threads, the first $5$ bits of an address define the thread in which transactions originating from this address can be included.

\begin{definition}
Given the sets of possible addresses $\mathcal{A}$ and transactions $\mathcal{T}$, a \emph{sharding} function $\mathcal{S}$ uniformly assigns any address $a\in\mathcal{A}$ to a particular thread $\mathcal{S}(a)=\tau\in[0,T-1]$, and any transaction $\mathrm{tx}\in\mathcal{T}$ to the thread assigned to the transaction's emitter address. The transaction $\mathrm{tx}$ can only be included in a block of thread $\mathcal{S}(\mathrm{tx})$, and can only reduce the balance of addresses assigned to this thread.
\end{definition}

Transaction sharding ensures that transactions in a block are compatible with transactions in blocks from other threads as they can't spend the same coins. 
We stress that this restriction only applies to spending, and transactions can send coins towards any address, regardless of the thread it is assigned to. Transactions in a thread are regularly taken into account in blocks of other threads through parent links, so that no further cross-shard communication is required.

\subsection{Sybil-Resistant Selection}
\label{sybil}

In a decentralized network, nodes can join and contribute without permission.
To control the rate of execution of transactions, nodes are regularly selected by the protocol to produce blocks of transactions with a limited size.
To prevent malicious actors from spawning an arbitrary large number of nodes (which is called a Sybil attack), and create too many blocks, the selection mechanism needs to rely on a proof of ownership of a resource.
The two main Sybil-resistant selection mechanisms used in current blockchains are Proof-of-Work and Proof-of-Stake.
Proof-of-Work \cite{nakamoto2008bitcoin} selects random nodes depending on their use of computational power, while Proof-of-Stake \cite{goodman2014tezos} selects them based on the amount of coins they hold.

The Blockclique protocol can use any Sybil-resistant selection mechanism that explicitly selects a node to create each block and endorsement. A node must know in advance, but not be able to choose in which threads it should produce its next blocks and endorsements.
For security reasons (see Sec. \ref{secret}), the selection mechanism must take into account the resources of nodes with some time delay $K$, called the resource snapshot delay.

\begin{definition} 
Let $\mathcal{N}$ be the set of nodes in the peer-to-peer network.
A \emph{$K$-Sybil-resistant selection} is a random oracle $\mathcal{S}~:~[0,1]^*~\rightarrow~\mathcal{N}$, accessible to all nodes, with a non-uniform non-stationary output distribution on $\mathcal{N}$. Its distribution in two incompatible cliques must stay the same for at least $K$ seconds after the first incompatible block between those cliques.
\end{definition}

The oracle models a random selection of nodes, which takes as input a string of bytes identifying a particular block slot or endorsement slot, and deterministically selects a node allowed to produce this block or endorsement.
All nodes consult this oracle to check when they are selected for a given slot or endorsement, and to verify that other nodes where allowed to create a given block or endorsement.
If no valid block is produced for a given block slot, the block slot remains empty, which corresponds to a block miss (see slot $s^1_1$ in Fig.~\ref{example1}(a)). Similarly, if no valid endorsement is produced for a given endorsement slot, the endorsement slot remains empty.

A Proof-of-Work mechanism like the one of Bitcoin is not directly adaptable to the Blockclique architecture.
Indeed, nodes could decide in which thread they produce blocks, or that thread could be chosen from the block hash as in the (non-sharded) OHIE protocol \cite{yu2018ohie}, in which case nodes do not know in advance which (sharded) transactions to include.
However, a Proof-of-Work mechanism could be used to generate identities, as in ELASTICO \cite{luu2016secure}, and those identities could then be randomly selected to produce blocks and endorsements.

Proof-of-Stake mechanisms like Tezos \cite{tezosdoc} are readily transferable to the Blockclique protocol.
Nodes that register to be stakers are randomly selected to produce blocks, with a probability proportional to their balance.
In the Tezos protocol, the stake snapshot delay $K$ is set to approximately $3$ weeks.
Moreover, a seed is computed in each cycle from bytes included in blocks by their producers, and is used to select pseudo-randomly the producers of a later cycle.  
Various implementations of random generator seeding can be considered \cite{pvss,boneh2018survey}.



In Blockclique, each block has a scalar fitness value, which measures the fraction of resources required for the creation of the block through Sybil-resistant selection.
The fitness $f(b)$ of a block $b$ is defined as the total number of selected addresses that successfully participated in the creation of the block:
\begin{equation}
    f\left(b\right) = 1 + e
    \label{eq:fitness}
\end{equation}
The scalar $1$ acknowledges successful block creation and inclusion, and $e$ is the number of endorsements successfully produced and included among the $E$ endorsement slots of block $b$.
This fitness value is used by the consensus rule to determine the set of executed transactions (see Sec. \ref{rule}).

\begin{figure*}[t]
\centering
\subfloat[Thread incompatibility]{\includegraphics[scale=0.6]{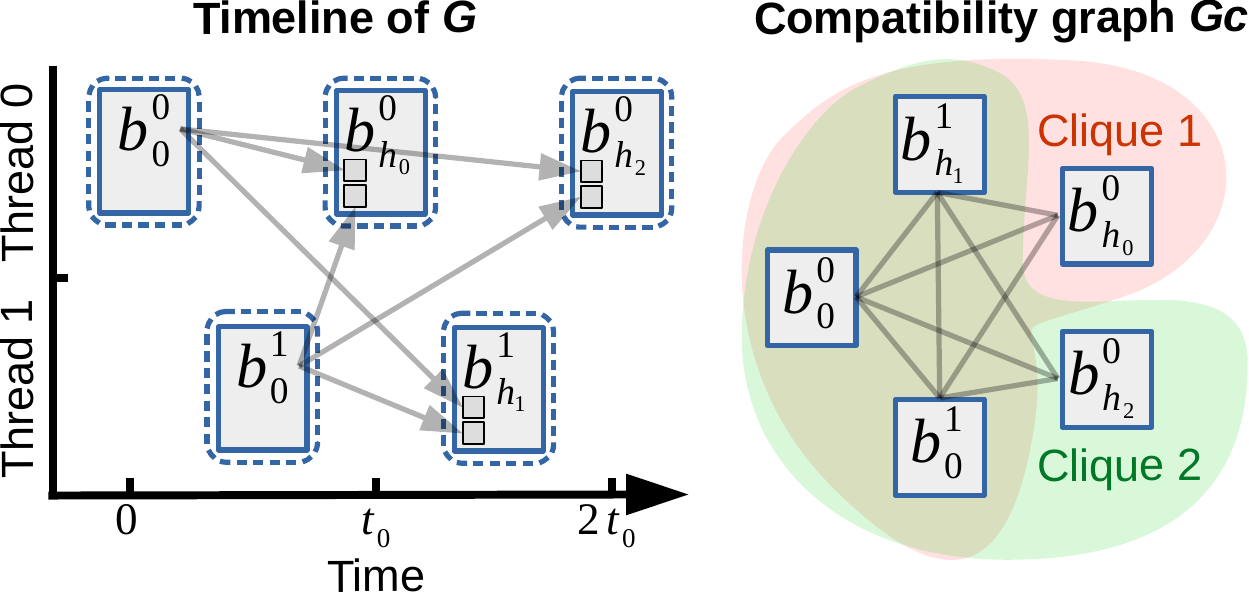}}
\hspace{40pt}
\subfloat[Grandpa incompatibility]{\includegraphics[scale=0.6]{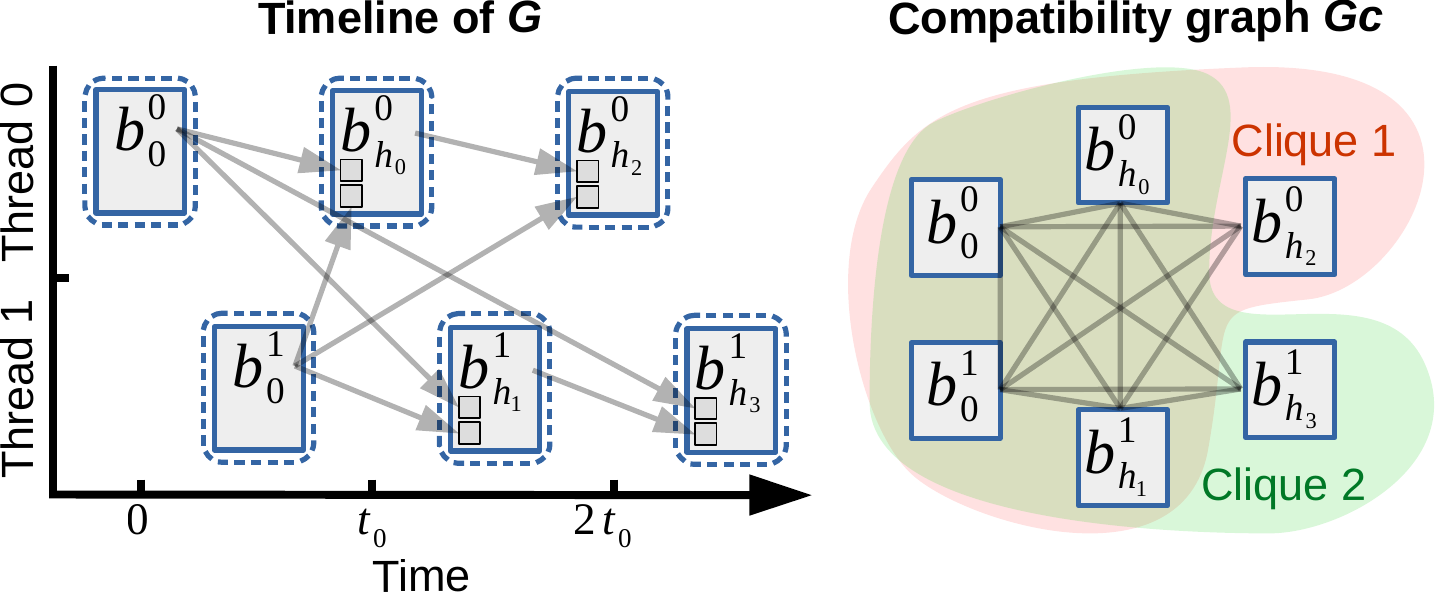}}
\caption{From Multithreaded block DAG to compatibility graph. (a) Blocks $b_{h_0}^0$ and $b_{h_2}^0$ from thread $0$ reference the same parent in thread $0$: they are thread-incompatible. (b) Block $b_{h_2}^0$ references the grand-parent of block $b_{h_3}^1$ in thread $1$ and block $b_{h_3}^1$ references the grand-parent of block $b_{h_2}^0$ in thread $0$: the two blocks are grandpa-incompatible. In both cases there are two maximal cliques of compatible blocks (red and green filled areas).}
\label{example_rule}
\end{figure*}

\subsection{Incentives: Rewards and Penalties}
\label{incentives}

In order to motivate nodes to participate in the consensus with the behavior specified by the protocol, an incentive model provides rewards for appropriate behavior and penalties for deviations from the protocol.

\begin{definition}
Given the set of addresses $\mathcal{A}$, a multithreaded block DAG $G$ and a block $b_{h}^{\tau}$, an \emph{Incentive model} $\mathcal{I}(G,b_{h}^{\tau})$ is a function assigning a reward and penalty to the addresses of all nodes: $\mathcal{I}(G,b_{h}^{\tau}): \mathcal{A} \rightarrow \mathbb R_{\ge 0}^2$.
\end{definition}

We assume that the creation of blocks is rewarded by newly created coins.
The block reward also contains a constant amount per included endorsement, shared between the block producer, the endorser, and the producer of the endorsed block, which motivates block creation as well as endorsement creation and inclusion.
The inclusion of transactions in blocks is rewarded by the fees from all included transactions. 

In order to prevent block producers from creating or endorsing multiple incompatible blocks in the same slot, we assume that the incentive model gives penalties to addresses involved in this misbehavior.
For instance, in Tezos those penalties are implemented by requiring block and endorsement producers to deposit an amount of coins that they can't withdraw for a given time~\cite{tezosdoc}.
Any node can produce a denunciation transaction containing the proof that a same address has produced or endorsed incompatible blocks at the same block slot.
A denunciation included in a block causes a coin penalty to the offender, taken from its deposit, half of which is destroyed, and half of which is transferred to the block creator.

In the context of the Blockclique architecture, transaction sharding requires that the offender address belongs to the same thread as the block in which the denunciation is included, because the offender is ``spending" the penalty.


\subsection{Consensus Rule}
\label{rule}

When a node receives a block from its peers, it checks that the block is valid, and uses a consensus rule to decide which valid blocks should be taken into account.
The intuitions behind the blockclique consensus rule are the following: on the one hand, each thread behaves like a standard blockchain so that two blocks in the same thread can't share the same parent in that thread (thread incompatibility), and on the other hand, rather than acting as if threads were independent, nodes should take into account blocks found in other threads (grandpa incompatibility).

\subsubsection{Compatibility Graph}

Let the predicate $\texttt{Path}_\tau(G, b_{h_1}^{\tau}, b_{h_2}^{\tau})$ be true if there is a directed path in the multithreaded block graph $G$ going from $b_{h_1}^{\tau}$ to $b_{h_2}^{\tau}$ through blocks of thread $\tau$ only, or if $b_{h_1}^{\tau} = b_{h_2}^{\tau}$. This predicate indicates whether or not $b_{h_1}^{\tau}$ is an ancestor in $\tau$ of (or is equal to) $b_{h_2}^{\tau}$.

We define the thread incompatibility graph $G_{TI}$ as the graph with one node per valid block, and an undirected edge between two blocks $b_{h_1}^{\tau_1}$ and $b_{h_2}^{\tau_2}$ only if the two blocks are non-genesis blocks in the same thread and have the same parent in their thread: 
\begin{multline}
G_{TI}(b_{h_1}^{\tau_1}, b_{h_2}^{\tau_2}) := \big[b_{h_1}^{\tau_1} \neq b_0^{\tau_1}\big] ~\texttt{and}~ \big[b_{h_2}^{\tau_2} \neq b_0^{\tau_2}\big] ~\texttt{and}~  \\ 
\hspace{1.5cm} \big[\tau_1 = \tau_2\big] ~\texttt{and}~ \big[P(b_{h_1}^{\tau_1}, \tau_1) = P(b_{h_2}^{\tau_2}, \tau_2)\big] \hspace{0.5cm}
\end{multline}
Fig.~\ref{example_rule}(a) shows an example of a thread incompatibility between two blocks.

We define the grandpa incompatibility graph $G_{GPI}$ as the graph with one node per valid block, and an undirected edge between two blocks $b_{h_1}^{\tau_1}$ and $b_{h_2}^{\tau_2}$ if 
the parent of block $b_{h_2}^{\tau_2}$ in thread $\tau_1$ is not the parent of $b_{h_1}^{\tau_1}$ nor one of its descendants in $\tau_1$, and the parent of block $b_{h_1}^{\tau_1}$ in thread $\tau_2$ is not the parent of $b_{h_2}^{\tau_2}$ nor one of its descendants in $\tau_2$:
\begin{multline}
\label{eq:gpi}
G_{GPI}(b_{h_1}^{\tau_1}, b_{h_2}^{\tau_2}) := \big[b_{h_1}^{\tau_1} \neq b_0^{\tau_1}\big] ~\texttt{and}~ \big[b_{h_2}^{\tau_2} \neq b_0^{\tau_2}\big] \\
\hspace{1cm} \texttt{and}~ \big[\texttt{not Path}_{\tau_1}(G, P(b_{h_1}^{\tau_1}, \tau_1), P(b_{h_2}^{\tau_2}, \tau_1))\big] \hspace{1cm} \\
\texttt{and}~ \big[\texttt{not Path}_{\tau_2}(G, P(b_{h_2}^{\tau_2}, \tau_2), P(b_{h_1}^{\tau_1}, \tau_2))\big]~~
\end{multline}
Grandpa incompatibility is a topological way of expressing that a block $b_h^\tau$ in thread $\tau$ should not be included if it does not take into account blocks that were found in other threads before the time when $P(b_h^\tau, \tau)$ was found, but without checking block timestamps that can be inaccurate or manipulated.
Fig.~\ref{example_rule}(c) shows an example of grandpa incompatibility between two blocks.

Using thread and grandpa incompatibility graphs $G_{TI}$ and $G_{GPI}$, we define the compatibility graph $G_C$ as the graph with one node per valid block, and an undirected edge between two blocks $b_{h_1}^{\tau_1}$ and $b_{h_2}^{\tau_2}$ if the two blocks are not thread nor grandpa incompatible, and $b_{h_1}^{\tau_1}$ is compatible with the parents of $b_{h_2}^{\tau_2}$, and $b_{h_2}^{\tau_2}$ is compatible with the parents of $b_{h_1}^{\tau_1}$:
\begin{multline}
\hspace{-0.32cm} G_C(b_{h_1}^{\tau_1}, b_{h_2}^{\tau_2}) := \\
\hspace{-1.5cm} \big[\texttt{not}~ G_{TI}(b_{h_1}^{\tau_1}, b_{h_2}^{\tau_2})\big] ~\texttt{and}~ \big[\texttt{not}~ G_{GPI}(b_{h_1}^{\tau_1}, b_{h_2}^{\tau_2})\big] \\ 
\hspace{-0.3cm} \texttt{and}~ \Big[\big[b_{h_2}^{\tau_2} = b_0^{\tau_2}\big] ~\texttt{or}~\big[G_C(b_{h_1}^{\tau_1}, P(b_{h_2}^{\tau_2}, \tau)) ~\texttt{for all}~ \tau\big]\Big] \\
 \texttt{and}~ \Big[\big[b_{h_1}^{\tau_1} = b_0^{\tau_1}\big] ~\texttt{or}~\big[G_C(P(b_{h_1}^{\tau_1}, \tau), b_{h_2}^{\tau_2}) ~\texttt{for all}~ \tau\big]\Big] 
\end{multline}
$G_C$ therefore links mutually compatible blocks, and blocks that reference mutually incompatible parents are discarded.
Figs.~\ref{example_rule}(b, d) show the $G_C$ graphs corresponding to the incompatibilities illustrated in Figs.~\ref{example_rule}(a, c).

The definition of $G_C$ is recursive: $G_C$ is built incrementally following a topological order of $G$, by processing a block as soon as all its parents have been received and processed.

\subsubsection{Best Clique of Compatible Blocks}
Let \texttt{cliques}$(G_C)$ be the set of maximal cliques of compatible blocks: the set of subsets $C$ of $G_C$ so that every two distinct blocks of $C$ are adjacent in $G_C$ and the addition of any other block from $G_C$ to $C$ breaks this property.
In the remainder of the paper, the term ``clique" refers to a maximal clique of compatible blocks.

The blockclique consensus rule states that the best clique, that nodes should extend, is called the \emph{blockclique} and is the clique of compatible blocks of maximum total block fitness:
\begin{multline}
\hspace{-0.2cm}\texttt{blockclique}(G) := \underset{C~\in~\texttt{cliques}(G_C)}{\mathrm{arg\,max}} ~ \bigg[\,\,\sum_{b \in C}f(b)\,\,\bigg]~~
\end{multline}
If two cliques have the same total fitness, the clique with the smallest arbitrary-precision sum of the hashes of the blocks it contains is preferred.

\subsubsection{Incremental Compatibility Graph and Finality}
\label{incremental}

As finding the maximal cliques of a graph is NP-hard \cite{karp1972reducibility}, the blockclique of the whole compatibility graph $G_C$ cannot be efficiently computed once $G_C$ contains more than a few hundred blocks.
Thus, an incremental rule for recomputing cliques using only the most recent blocks is required.
We define $G_C^\mathrm{head}$ as a minimal version of $G_C$ from which blocks that 
are considered final (forever part of the blockclique) or stale (forever incompatible with the blockclique) have been removed. 
$G_C^\mathrm{head}$ is kept in memory and updated incrementally.

A block $b_{h}^{\tau}$ is considered stale if it is included only in cliques of $G_C^\mathrm{head}$ that have a total fitness lower than the fitness of the blockclique minus a constant $\Delta_f^0$. 
Any new block referring to stale parents is considered stale. 
A block $b_{h}^{\tau}$ is considered final if it is included in all maximal cliques of $G_C^\mathrm{head}$ and included in at least one clique where the descendants of $b_{h}^{\tau}$ cumulate a total fitness greater than $\Delta_f^0$.

We define the threshold fitness difference $\Delta_f^0 = F(E+1)$, where $E$ the number of endorsement slots per block, and $F$ is a finality parameter that can be seen as the number of blocks by which an alternative clique can be shorter than the blockclique before its blocks may be discarded as stale. 

\begin{definition} 
Given a current incremental compatibility graph $G_C^\mathrm{head}$ and a new block $b_{h}^{\tau}$, a \emph{$(F,E)$-Nakamoto} consensus rule outputs a set of final and stale blocks to be removed from $G_C^\mathrm{head}$, and the blockclique to be considered.
\end{definition}

\section{Security}\label{sec:security}


\subsection{Threat Model}

We consider a network composed of honest and Byzantine nodes.
Honest nodes follow the Blockclique protocol, while Byzantine (``attacker") nodes seek to disturb its functioning for their benefit or even at their own cost. Byzantine nodes hold a proportion $\beta$ of the total resource, and honest nodes own $1-\beta$.
Furthermore, we generalize the behavior of honest nodes by assuming that they may not be perfect: they miss block creation and endorsement opportunities with a probability $\mu$.
We define $\gamma = (1-\beta)(1-\mu)$ as the proportion of the total resource that is in active use by honest nodes.
Attackers are assumed to be able to delay the propagation of messages between honest nodes by a maximum time of $\delta$ seconds, so that a block or endorsement created by a honest node is broadcast to all other honest nodes before the delay $\delta$.

\subsection{Attack Surface}
\label{surface}

Attackers are fully coordinated and always behave in the optimal way to perform a given attack. 
They can choose to honor or miss block creation and endorsement opportunities in the blockclique and/or any alternative cliques meant to attack the blockclique. 
When creating a block, they choose which transactions and endorsements to include (if any).

In Blockclique, as in blockchains with a Nakamoto consensus, consensus emerges through block creation.
In Bitcoin, blocks are never definitely confirmed and the confirmation status of a block increases with the number of blocks that are appended to it.
In Blockclique however, the status of a block is eventually settled: it either becomes part of the history (final) or is discarded (stale). Transactions in final blocks are considered as perpetually executed by honest actors, while transactions only within a stale block are discarded. 
The boundary between settled and unsettled blocks is controlled by the finality parameter $F$.

In blockchains, attackers can try to re-organize the blockchain by extending an alternative branch of the block tree, for instance to attempt a double-spend.
In Blockclique, the introduction of the finality parameter modifies the mechanisms and the consequences of such attacks.
Attackers can branch off the current blockclique and extend an alternative clique from three possible levels.

Attackers can branch off a recent block that is universally seen as unsettled.
If they succeed in extending the alternative clique and overtaking the fitness of the blockclique, honest nodes switch to the alternative clique.
However, this does not change the finality status of any block according to any honest node, and therefore has no consequences on the finality of transactions.

Alternatively, attackers can branch off a block that is settled according to some but not all honest nodes due to network delays. If they succeed in overtaking the blockclique, this can lead to a network fork.
We call this a \textit{finality fork} attack, and study it in Sec.~\ref{finality_fork}.

Finally, attackers can also branch off an old block, known by all honest nodes under the same status, final or stale.
If they succeed in overtaking the blockclique, honest nodes are not affected as they consider all the descendants of the block as stale.
However, new nodes joining the network and simply choosing the clique of highest fitness are vulnerable.
We study this attack on newcomers in Sec.~\ref{secret}.

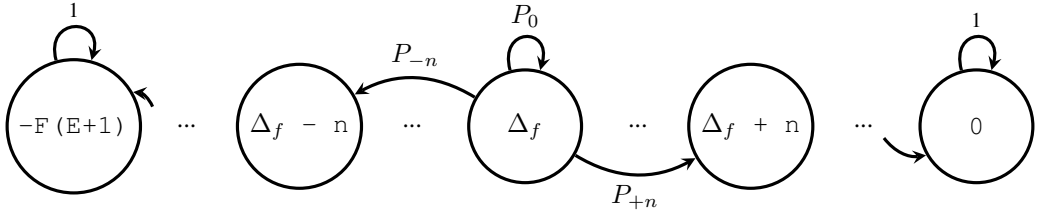
\begin{figure*}[h]
\centering
\begin{tikzpicture}[scale=1, transform shape, >=stealth, node distance=1.5cm, very thick]
	\node[draw, shape=circle,minimum size=1.5cm] (0) {\texttt{-F(E+1)}}; 
	\node[shape=circle,minimum size=1cm, right of=0] (1) {...};  
	\node[draw, shape=circle,minimum size=1.5cm, right of=1] (2) {$\Delta_f$\texttt{ - n}}; 
	\node[xshift=1.5cm, at=(2)] (3) {...};  
	\node[draw, shape=circle,minimum size=1.5cm, right of=3] (4) {$\Delta_f$}; 
	\node[xshift=1.5cm, at=(4)] (5) {...};  
	\node[draw, shape=circle,minimum size=1.5cm, right of=5] (6) {$\Delta_f$\texttt{ + n}}; 
	\node[xshift=1.5cm, at=(6)] (7) {...};  
	\node[draw, shape=circle,minimum size=1.5cm, right of=7] (8) {\texttt{0}};
	\path (0) edge [->, out=105, in=75, min distance=0.6cm] node[above] {\footnotesize 1} (0)
	(8) edge [->, out=105, in=75, min distance=0.6cm] node[above] {\footnotesize 1} (8) 
	(4) edge [->, out=105, in=75, min distance=0.6cm] node[above] {$P_0$} (4) 
	(4) edge [->, bend right] node[below] {$P_{+n}$} (6) 
	(4) edge [->, bend right] node[above] {$P_{-n}$} (2) 
	(7) edge [->, bend right] node[above] {} (8) 
	(1) edge [->, bend right] node[above] {} (0) 
	;
\end{tikzpicture} 
\caption{Markov chain representing the fitness difference $\Delta_f$ between an alternative attack clique and the blockclique.}
\label{pos_markov}
\end{figure*}

\subsection{A Markov Chain Model of Alternative Cliques}
\label{mcmodel}

In this section, we model the short-term evolution of the fitness of alternative cliques.
We consider a honest reference node receiving blocks (not necessarily in their order of creation), and we model the fitness difference $\Delta_f$ between the clique extended by an attacker and the blockclique extended by honest nodes through a Markov chain of state $\Delta_f$, as observed and processed by the reference node.
Fig. \ref{pos_markov} shows the Markov chain model representing the fitness difference $\Delta_f$.

We assume that the attack occurs within the resource snapshot delay $K$, so that resources and random selection results are the same in both cliques. 

We also assume that the attacker arbitrarily delays the transmission of messages (blocks and endorsements) up to a time delay $\delta < \frac{t_0}{2}$.
With this assumption, honest nodes always receive a created block $b_h^\tau$ before a delay $\frac{t_0}{2}$, then create and broadcast endorsements of $b_h^\tau$, so that the next block producer in thread $\tau$ receives the block and its endorsements before the time $t_0$ when it is supposed to create and broadcast the next block.
Honest nodes thus never create incompatible blocks.
Also, in case the Sybil-resistant mechanism modifies the selection probabilities during the attack, we conservatively consider $\beta$ to be the maximum proportion of resources the attacker reaches during the attack.

The model considers the worst-case scenario in which the attacker never misses block creation nor endorsement opportunities in the attack clique and always misses in the blockclique, while honest nodes miss in the blockclique with a probability $\mu$ and always miss in the attack clique.

If the fitness difference $\Delta_f$ reaches $-F(E+1)$, where $F$ is the finality parameter and $E$ the number of endorsement slots per block, the attack fails.
On the contrary, if $\Delta_f$ reaches $0$, the attack clique overtakes the blockclique and the attack succeeds.
The states $-F(E+1)$ and $0$ are therefore the two absorbing states of the Markov chain, which constrains states within $-F(E+1) \leq \Delta_f \leq 0$.

When the honest node receives a block created by the attacker, $\Delta_f$ transitions forward to $\Delta_f + n$, with $1 \leq n \leq E+1$ depending on the number of endorsements the attacker was selected to create for the previous block. If $\Delta_f + n \geq 0$, the Markov chain enters and remains in the attack success state $\Delta_f=0$.
The probability $P_{+n}$ of such a $n$-point forward jump is the probability that the attacker is selected for the creation of one block and for $n-1$ endorsements in that block out of $E$ slots (independent draws of a binomial law):
\begin{equation}
    P_{+n} = \beta~\binom{E}{n-1} \beta^{n-1} (1-\beta)^{E-(n-1)} 
\end{equation}

Similarly, when the honest node receives a block created by a honest node, $\Delta_f$ transitions backwards to $\Delta_f - n$ with $1 \leq n \leq E+1$. If $\Delta_f - n \leq -F(E+1)$, the Markov chain enters and remains in the attack failure state $\Delta_f = -F(E+1)$.
The probability $P_{-n}$ of a $n$-point backward jump is the probability that a honest node is selected for the creation of one block and does not miss it and that despite endorsement misses, honest nodes produce exactly $n-1$ endorsements out of $E$:
\begin{equation}
    P_{-n} = \gamma~\binom{E}{n-1} \gamma^{n-1} (1-\gamma)^{E-(n-1)}~~~
\end{equation}

If a selected honest node misses block creation, the state of the Markov chain does not change, which happens with probability $P_0 = (1-\beta)\,\mu$.

From $P_{+n}$, $P_{-n}$ and $P_{0}$, we deduce the matrix of transition probabilities from any state to any other.
Standard techniques for absorbing Markov chains \cite{grinstead1997chapter} provide ways to numerically compute the probability that the attack clique eventually overtakes the blockclique depending on the initial state. 

\begin{figure*}[t]
\centering
\subfloat[$E=0$, $\mu=1\%$]{\includegraphics[width=0.82\columnwidth]{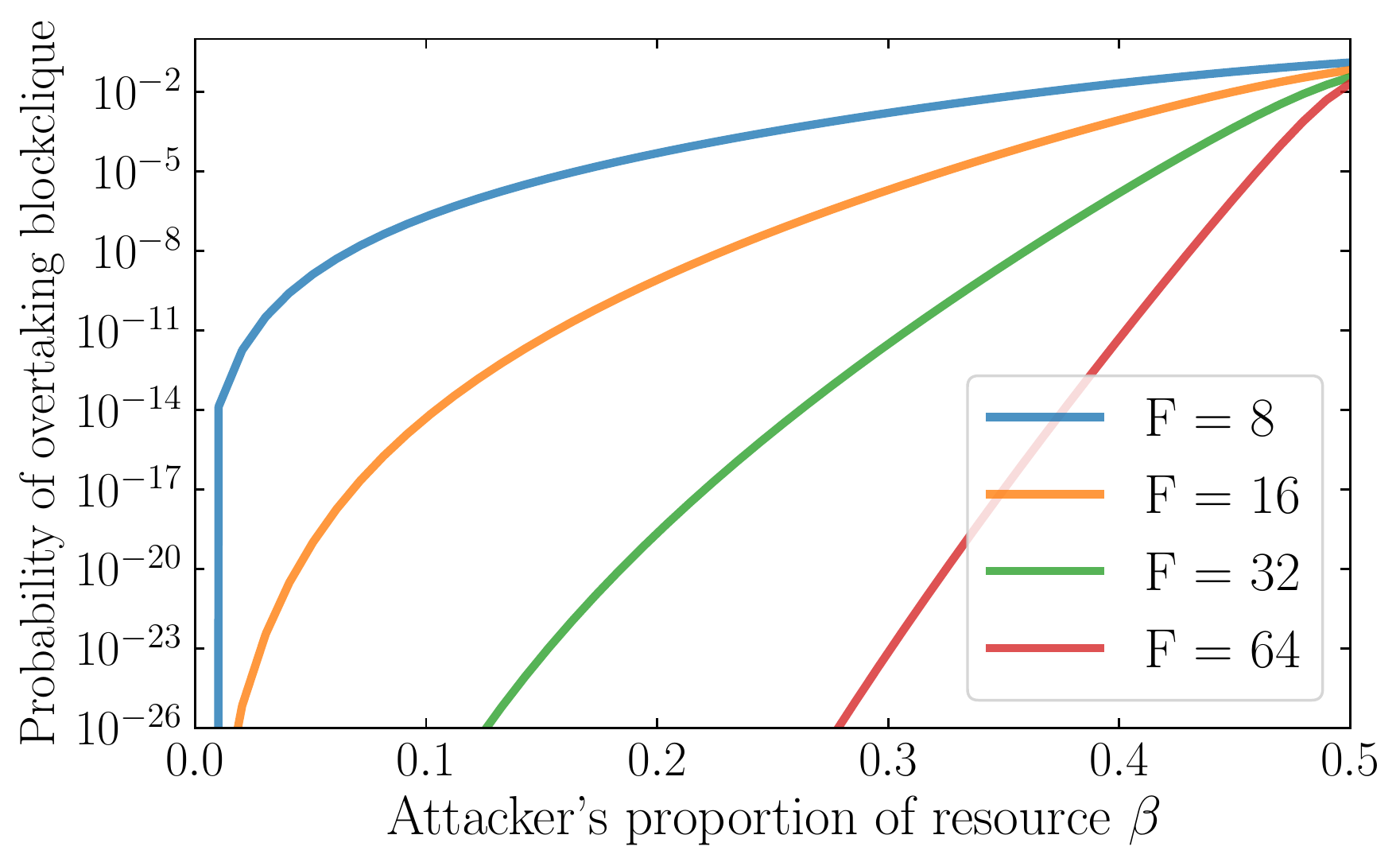}}\hspace{1.5cm}
\subfloat[$F=64$, $\mu=1\%$]{\includegraphics[width=0.82\columnwidth]{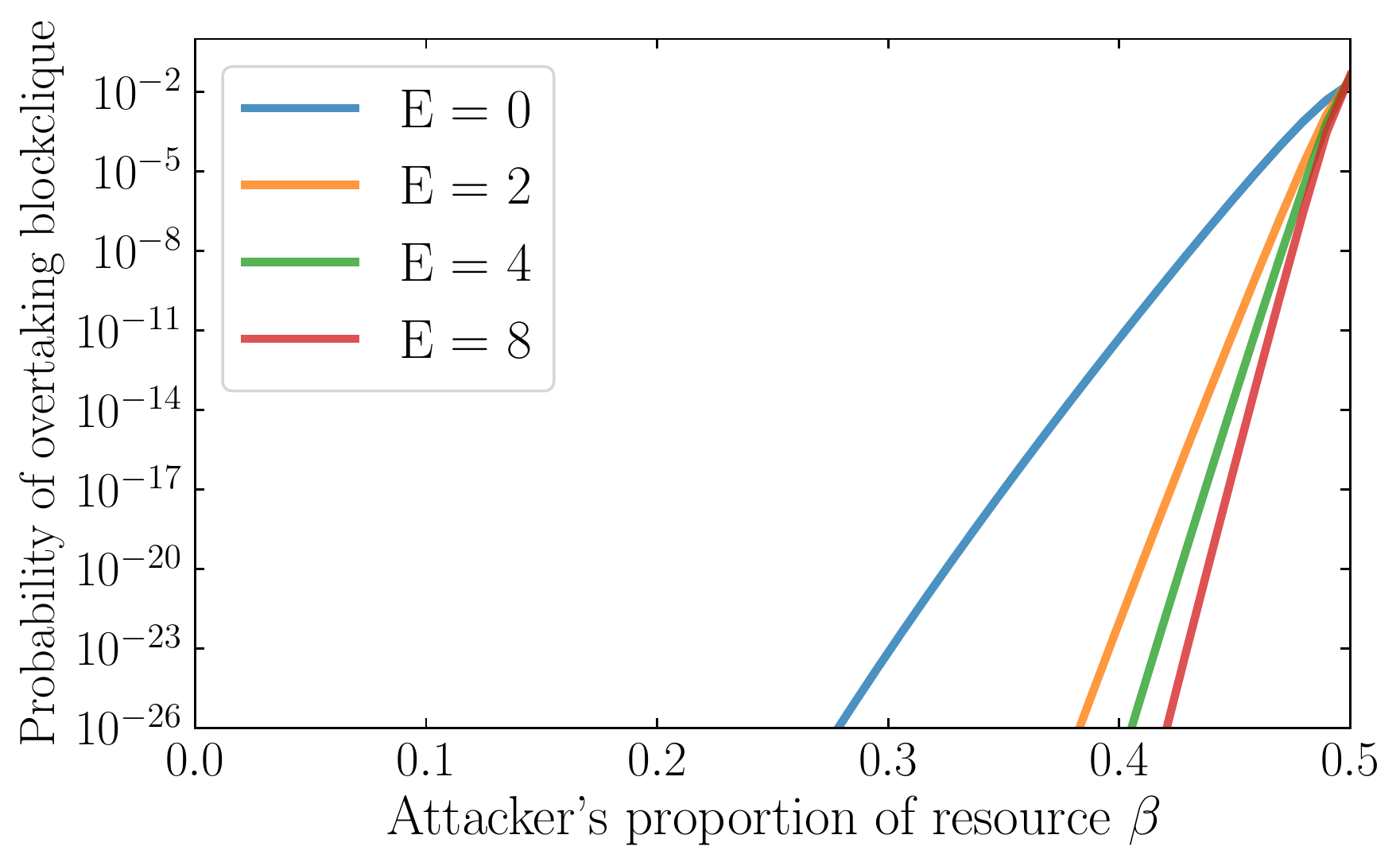}}
\caption{Probabilities of finality fork attack success as functions of the proportion of staking power $\beta$ held by attackers, depending on $F$ and $E$.}
\label{markov_res_a}
\end{figure*}

\subsection{Attacking Block Finality Consensus}
\label{finality_fork}

Decentralized currencies require strong guarantees that transactions labeled as \emph{final} are perpetual and may not be canceled in the future.
However, due to network propagation delays, consensus on whether a block is settled is not reached instantly. For a short time, some nodes may have settled a block as stale or final, while others haven't yet.
Powerful attackers can coordinate to timely extend an alternative clique from a block incompatible with this one. If the fitness of the attack clique overcomes the current blockclique, part of the honest nodes switch to the attack clique as their new blockclique while others only keep the initial blockclique and reject the attack clique as stale, resulting in a permanent network fork.

We consider two cliques: the current blockclique being attacked, and an alternative attack clique incompatible with the current blockclique being extended by the attackers.
In the current blockclique, honest nodes produce blocks and endorsements with a miss probability $\mu$, while attackers miss all block creations and endorsements to slow down the increase of the fitness of the current blockclique. In the attack clique, honest nodes are not present and miss all their block creation and endorsement opportunities, while attackers never miss block creation and endorsements to maximize the fitness increase of the attack clique.

In order to evaluate the success probability of a finality fork attack, we use the Markov chain of Sec.~\ref{mcmodel} modeling the evolution of the fitness difference $\Delta_f$ between the attack clique and the current blockclique at every new slot. 
The attack starts when the attack clique is on the verge of reaching the finality threshold ($\Delta_f = -(F-1)(E+1)$), fails if it does (reaching absorbing state $\Delta_f = -F(E+1)$), and succeeds if its fitness catches up with the current blockclique (reaches absorbing state $\Delta_f = 0$). The following Lemma shows that in the case $E = 0$ the success probability of a finality fork attack drops exponentially in $F$.



\begin{lemma}[Success of a finality fork attack]
\label{lemma_finality}
Assuming $\beta < \gamma$, $\delta < \frac{t_0}{2}$ and $E=0$, the probability $p$ of a finality fork attack success is $$p = \frac{\frac{\gamma}{\beta} -1}{\left( \frac{\gamma}{\beta} \right)^F - 1} \underset{F \to \infty}{\sim} \left(\frac{\gamma}{\beta} - 1\right)\left( \frac{\beta}{\gamma} \right)^F$$
\end{lemma}

\begin{proof}

The finality fork attack is successful if the attack clique, starting with a fitness difference $\Delta_f = -F+1$, reaches a fitness difference of $0$. This problem is analogous to the two barrier ruin problem for Bernoulli random walks. A standard result for random walks with drift (see \cite{soren2010ruin}, p. 297) shows that, starting from state $\Delta_f$, the probability of reaching the absorbing state $\Delta_f = 0$ is given by
\begin{equation}
p(\Delta_f) = 1 - \frac{\left( \frac{\gamma}{\beta} \right)^F - \left( \frac{\gamma}{\beta} \right)^{F- \Delta_f}}{\left( \frac{\gamma}{\beta} \right)^F - 1}
\end{equation}
A finality fork attack starts at $\Delta_f = -F + 1$, in which case
\begin{equation}
p(-F + 1) = 1 - \frac{\left( \frac{\gamma}{\beta} \right)^F - \frac{\gamma}{\beta}}{\left( \frac{\gamma}{\beta} \right)^F - 1} = \frac{\frac{\gamma}{\beta} -1}{\left( \frac{\gamma}{\beta} \right)^F - 1}
\end{equation}
\end{proof}

The case $E > 0$ is more involved to analyze formally but a numerical computation (see Sec. \ref{mcmodel}) shows that increasing $E$ improves the security of the protocol against the finality fork attack. Figure \ref{markov_res_a} shows example values of the attack success probability computed for different $\beta$, $F$ and $E$.
For instance, with $F = 64$, $E = 0$, $\mu = 1\%$ and $\beta = 45\%$, the success probability of a finality fork attack is about $10^{-6}$, while with $E=8$, it becomes about $10^{-16}$.

For the Markov chain hypotheses to hold, the resource snapshot delay $K$ must be longer than the possible duration of an attack. The following Lemma show that the probability that the duration of an attack last more than $n$ slots decrease exponentially with $n$. It follows that it is possible to chose $K$ such that the Markov chain hypotheses hold except with probability as small as desired. The attack duration and its standard deviation increase with $\beta$. Numerical results show that for $\beta = 0.5$, $F = 64$ and $E = 8$ attacks reach an average duration of $410\,$ slots (s.d. of $\sigma=598$ slots).

\begin{lemma}[Duration of a finality fork attack]
\label{lemma_duration_finality}
The probability that a finality fork attack lasts more than $n$ slots decreases exponentially with $n$.
\end{lemma}

\begin{proof}
Let $t_a$ be the duration of the attack (either successful or not). Let $p_t = P(t_a > t)$ be the probability that the attack has not finished after a duration $t$. 
It suffices to find an upper bound on the probability that the attack has not stopped after $n=t \times T/t_0$ slots. 
One way for the attack to terminate is when a sequence of $n$ slots contains a subsequence of length $F(E+1)$ containing only forward jumps or only backward jumps. Such subsequences happen with probability $\beta^{F(E+1)}$ and $\gamma^{F(E+1)}$, respectively. 
Considering only non overlapping subsequences, shows that the probability that a sequence of length $n$ does not contain any such subsequence is bounded above by $q_t = (1 - \beta^{F(E+1)} - \gamma^{F(E+1)})^{\left\lfloor \frac{n}{F(E+1)} \right\rfloor}$. 
This bound $p_t \le q_t$ implies that $p_t$ decreases exponentially with $n.$
\end{proof}

\subsection{Consistency of the Blockclique Protocol}

In the context of blockchains, a protocol is said \emph{consistent} if it guarantees that all honest nodes eventually agree on the same set of final blocks \cite{kiffer2018better}.
The finality parameter $F$ in the Blockclique architecture implies a risk of a finality fork attack, which can be made arbitrarily small by increasing $F$ at the cost of longer transaction confirmation times.
Here, we formally establish this property by showing that if the sets of final blocks seen by two honest nodes are compatible at a time $t$, then they are still compatible at a time $t+r$ with high probability. The proof involves Lemma~\ref{lemma_finality}, as well as the Blockclique consensus rules.

\begin{theorem}[Consistency]
Consider a Blockclique protocol and a network such that $\beta < \gamma$, $\delta < \frac{t_0}{2}$ and $E=0$. 
Let $\mathcal{F}^t_1$ and $\mathcal{F}^t_2$ denote the sets of final blocks observed at time $t$ by nodes $n_1$ and $n_2$ respectively.
Given $t>0$ and $r>0$, if $\mathcal{F}^t_1$ is compatible with $\mathcal{F}^t_2$, then $\mathcal{F}^{t+r}_1$ is compatible with $\mathcal{F}^{t+r}_2$, except with a probability that drops exponentially in $F$.
\end{theorem}

\begin{proof}

Consider two honest nodes $n_1$ and $n_2$. 
As argued in Sec. \ref{surface}, the only strategy for an attacker starting an attack at time $s$ to make the sets of final blocks of $n_1$ and $n_2$ incompatible, is, given a block $b_{h}^{\tau}$ final at a time $s$ according to one of the two nodes, say $n_1$, but not yet settled according to $n_2$, to create a block $b_{h'}^{\tau'}$ incompatible with $b_{h}^{\tau}$, and try to overtake the blockclique with this alternative clique. 
By Lemma~\ref{lemma_finality}, the probability $p$ of success of this attack is $\exp(-\Omega(F))$.

If a finality fork attack is already happening at time $t$, then its probability to succeed before time $t+r$ is lower than $p$.
As the attacker may spawn other attacks one after the other independently, the probability that one of a maximum of $m=\lfloor\frac{r T}{t_0}\rfloor$ consecutive attacks starting between time $t$ and $t+r$ succeeds before time $t+r$ is lower than $1 - (1-p)^m$.
The probability that the attack started before time $t$ or any attack started between time $t$ and $t+r$ succeeds before time $t+r$ is thus $\exp(-\Omega(F))$.

Given the Nakamoto consensus rule defined in Sec.~\ref{incremental}, if a finality fork attack started at time $s$ with block $b_{h}^{\tau} \in \mathcal{F}^s_1$ and $\notin \mathcal{F}^s_2$ and incompatible block $b_{h'}^{\tau'}$ succeeds between times $t$ and $t+r$ according to $n_2$, then $b_{h'}^{\tau'} \in \mathcal{F}^{t+r}_2$.
In this case (of probability $\exp(-\Omega(F))$), $\mathcal{F}^{t+r}_1$ and $\mathcal{F}^{t+r}_2$ are not compatible. 
If all finality fork attacks started before time $t+r$ failed before time $t+r$, then given the stale block rule, the block $b_{h'}^{\tau'}$ and other blocks present in the attack clique and not in the blockclique become stale also according to $n_2$. In that case, as blocks produced by honest nodes are compatible with each other provided $\delta < \frac{t_0}{2}$, $\mathcal{F}^{t+r}_1$ is compatible with $\mathcal{F}^{t+r}_2$.
If a finality fork attack is still running at time $t+r$ as observed by node $n_2$, and previous attacks failed, then $\mathcal{F}^{t+r}_1$ is compatible with $\mathcal{F}^{t+r}_2$.

\end{proof}

\subsection{Attacking Honest Newcomers}
\label{secret}

Attacking new honest nodes joining the network involves creating and extending an alternative clique until its fitness becomes higher than the current blockclique's, while allowing it to become stale from the point of view of all existing honest nodes. 
Since newcomers are not aware of the finality status of attack clique's blocks, they consider the attack clique as the best clique, which prevents them from synchronizing with other honest actors on the honest clique and causes a network fork.

After the start of the attack clique, and before a change in resource snapshot, safety is ensured when the attack clique's fitness grows slower on average than the current blockclique's.

Given a probability $p$ of block and endorsement creation, 
the expected number of endorsements per block is $pE$, and therefore the expected fitness of a block is $1+pE$.
As the probability of block creation is $p$, the expected fitness increase of a clique per block slot is $p(1+pE)$.

The probability of successful block and endorsement inclusion is $\beta$ in the attack clique, and $\gamma$ in the current blockclique. The expected fitness increase per block slot is therefore $\beta(1+\beta E)$ in the attack clique and $\gamma(1+\gamma E)$ in the honest clique. As a result, the fitness of the attack clique grows on average slower than the one of the honest clique if $\beta < \gamma$.



However, beyond the resource snapshot delay $K$, resources are not guaranteed to be equal in the two cliques anymore.
Moreover, depending on the implementation of the Sybil-resistant selection mechanism, a node may get deactivated due to inactivity to avoid high miss rates.
In the attack clique, honest nodes get deactivated due to inactivity and the attackers become the sole block producers, causing the fitness of the attack clique to overcome the fitness of the current blockclique and the attack to succeed in the long run.

To prevent this attack, the client software downloaded by newcomers should include the hash of a recent final block of the blockclique (called checkpoint), allowing them to discard high-fitness attack cliques during bootstrap. 
The checkpoint must be more recent than the change in resource snapshot to ensure that no highest-fitness attack clique has started after the checkpoint as long as $\beta < \gamma$.

\begin{property}
If a new honest node is provided with a checkpoint more recent than a $K$-seconds resource snapshot delay, and if $\beta < \gamma$, then it is safe against attacks with alternative cliques branching off settled blocks.
\end{property}

For instance, with $\mu=1\%$, the safety condition becomes $\beta \leq 0.497$.

\subsection{Security Constraints on the Parameters}

Our analysis shows that for $F \geq 64$, $E \geq 8$ and $\delta < \frac{t_0}{2}$, the system is robust against architecture-based attacks as long as the attacker resource proportion $\beta$ is below $45\%$.
This particular limit assumes a miss rate of $\mu=1\%$ which is the current one in Tezos.
Furthermore, long term attacks are prevented by providing a recent block hash checkpoint to new nodes when they join the network, in addition to the client software and an IP list of bootstrap nodes.

\section{Performance}
\label{perf}

\begin{figure*}[t]
	\centering
	\subfloat{\includegraphics[height=4.8cm]{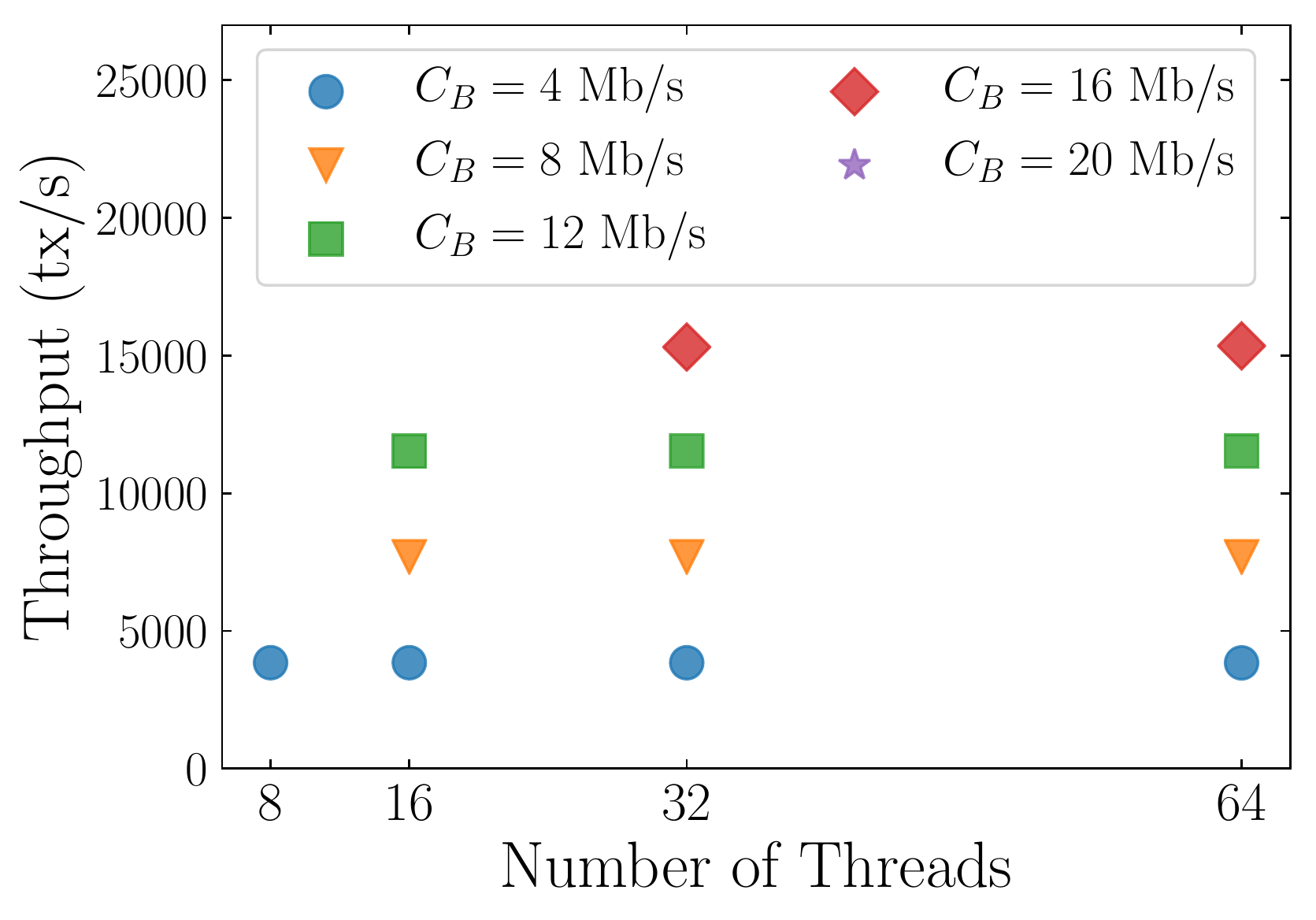}}\hspace{1.5cm}
	\subfloat{\includegraphics[height=4.8cm]{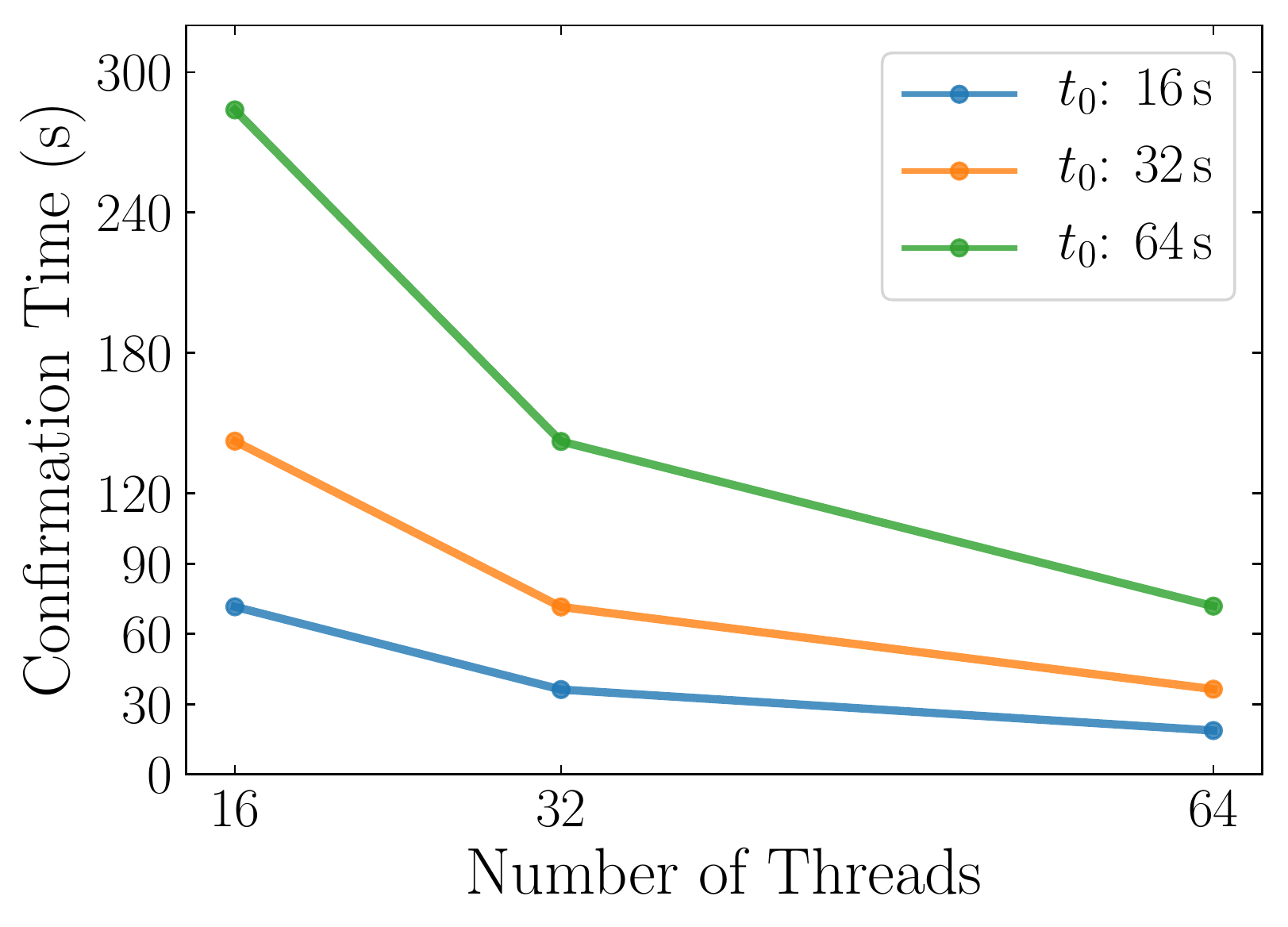}}
	\caption{
	Left: Transaction throughput depending on the number of threads $T$ and consensus bitrate $C_B$, with an inter-block interval $t_0=32$\,s.
	The average throughput over $10$ runs with different seeds is plotted only if all runs yielded a low stale rate ($< 2\%$).
	Right: Confirmation time, depending on $T$, $t_0$, with $C_B=12$\,Mb/s and $F=64$ blocks.
	}
	\label{res:pos_main1}
\end{figure*}

As in blockchains, performance in Blockclique is constrained by protocol parameters and network properties.
We evaluate the performance of Blockclique by simulating a peer-to-peer network of nodes transmitting and verifying blocks, and independently applying the consensus rules when they receive a block.
We estimate the highest consensus bitrate achievable at low stale rate by optimizing architecture parameters under various assumptions on network properties.

In the following case, network properties are chosen to match the estimates in Bitcoin and Ethereum \cite{gencer2018decentralization}: up to $N=4096$ nodes are randomly connected in a peer-to-peer network, with a median bandwidth $B = 32$\,Mb/s, and latency $L=100$\,ms between two connected nodes.
The number of threads $T$, the time between blocks $t_0$, and the block size $S_B$ define the consensus bitrate $C_B=\frac{T.S_B}{t_0}$.
Results show that parallelizing block creation in $T=32$ threads allows to safely reach a consensus bitrate $C_B=12$\,Mb/s in a network of $N=4096$ nodes, leading to a transaction throughput above $10000$\,tx/s.
Our implementation is open-source\footnote{\href{https://gitlab.com/blockclique/blockclique}{\small \texttt{gitlab.com/blockclique/blockclique}}}.

\subsection{Simulation Methods}

\subsubsection{Peer-to-peer Network}

The peer-to-peer network is generated as a directed graph of $N$ nodes with random connections between peers.
Each node has a particular upload bandwidth $b$ for sending blocks, randomly sampled at the beginning of the simulation between $\frac{1}{2}B$ and $\frac{3}{2}B$ where $B$ is the average upload bandwidth of all nodes.
Each node sends blocks one by one sequentially at the maximum speed of its upload bandwidth, and with a random latency depending on the destination node.
The latency between two nodes is sampled at the beginning of the experiment between $0$\,ms and $2 L$ where $L$ is the average latency between two nodes of the network.
Given its upload bandwidth $b$, a node is connected to a number  $\lfloor 4 b/B \rfloor$ of random successors.

When receiving a block, nodes verify the block and its transactions before forwarding it to their successors in the network graph.
Those verifications are simulated in the sense that we only consider a theoretical time needed for the node to verify the block and the transactions.
The block verification time is set to $50$\,ms, and the transaction verification time to $0.025$\,ms per transaction included in the block.
Each time a node receives a new block, its compatibility graph and blockclique are updated given its locally observed blocks, so that when creating a new block, it extends its local blockclique.



%

\subsubsection{Sybil-Resistance}	

Nodes are uniformly selected to produce blocks in particular threads at particular times, modeling a uniform resource distribution.
We implement this process by seeding a pseudo-random generator with the thread and slot numbers of a block before sampling a node that will have a right to produce a block in that slot of that thread.
A node produces a block in a thread as soon as the timestamp of the slot is reached by the computer clock.

\subsubsection{Blocks and Transactions}

Blocks are assumed to be full of transactions.
The size $S_{tx}$ of a transaction is set to $S_{tx}=1040$\,bits (transaction with one input and one output).
The maximum number of transactions per block is thus $\frac{S_B-S_H}{S_{tx}}$, where $S_B$ is the block size and $S_H$ the block header size.
We do not simulate endorsements for simulation time reasons, however their impact on transaction throughput is expected to be minimal as one endorsement is about the same size as a transaction.

\subsection{Results}

\subsubsection{Optimization of Blockclique Parameters}

In this section, we assume a network of $N=1024$ nodes, with median bandwidth $B=32$\,Mb/s, and median latency $L=100$\,ms between nodes.

We evaluate different architecture parameters to assess which consensus bitrate is viable in the range $C_B=4$\,Mb/s to $20$\,Mb/s, depending on the number of threads $T$, the inter-block time $t_0$ and a block size constrained by the other parameters: $S_B=\frac{C_B.t_0}{T}$.
In order to test the network under maximal load, nodes are assumed to produce all blocks ($\mu=0$).

Figure~\ref{res:pos_main1}(a) shows the average transaction throughput over $10$ runs with different seeds.
Whenever one of the $10$ runs shows a significant stale rate (more than the $2\%$ of Bitcoin), the corresponding point is not plotted.
We only report the case $t_0=32$\,s as other values ($t_0=16$ and $64$\,s) yielded the same results.

Our results show that the maximum consensus bitrate that can be achieved with a low stale rate increases with the number of threads up to $T=32$ threads, achieving $C_B=4$, $12$, $16$ and $16$\,Mb/s for $T=8$, $16$, $32$ and $64$ threads.
The corresponding transaction throughput is for instance $15307$\,tx/s with $T=32$ threads, $t_0=32$\,s, and a block size $S_B=16$\,Mb ($C_B=16$\,Mb/s).

In a separate experiment, we evaluate the transaction throughput when nodes miss a proportion $\mu$ of the blocks.
For $T=32$, $t_0=32$\,s, and $C_B=12$\,Mb/s, the resulting throughput is proportional to $(1-\mu)$: $11532$, $10342$, $9218$ and $8070$\,tx/s for $\mu=0$, $0.1$, $0.2$, and $0.3$.

Figure~\ref{res:pos_main1}(b) shows the measured transaction confirmation time as a function of the number of threads $T$ and the inter-block interval $t_0$, with a consensus bitrate $C_B=12$\,Mb/s and $F=64$ blocks.
The confirmation time is approximately the block finality time $\frac{F.t_0}{T}$ plus the time $t_{1/2}$ for a block to be broadcast to most of the nodes.
With $T=32$ threads, the confirmation time is $36$\,s, $72$\,s and $142$\,s for $t_0=16$\,s, $32$\,s and $64$\,s respectively. We measure network latencies $t_{1/2}=4$\,s, $7$\,s and $13$\,s respectively (average time to broadcast a block to half the network, when the block size is $S_B=6$, $12$ and $24$\,Mb).

\subsubsection{Influence of Network Properties}

In the previous section, we assumed a network with the following properties: $N=1024$, $B=32$\,Mb/s, and $L=100$\,ms, and studied the architecture performances depending on its parameters.
Here, we provide additional results evaluating the influence of network properties on the best achievable consensus bitrate.

When the number of nodes is increased to $N=4096$, keeping $B=32$\,Mb/s, $L=100$\,ms, $T=32$ threads and $t_0=32$\,s, the network reaches a viable consensus with a stale rate below $1\%$ up to $C_B=12$\,Mb/s, yielding an average transaction throughput of $11,500$\,tx/s.
When the median latency is varied from $L=50$ to $150$\,ms, keeping $N=1024$, $B=32$\,Mb/s, $T=32$ and $t_0=32$\,s, the network remains stable with $C_B=12$\,Mb/s.
Further simulations at a very low average bandwidth $B=4$\,Mb/s show that for $N=1024$, $L=100$\,ms, $T=32$ threads and $t_0=32$\,s, the network supports a consensus bitrate up to $C_B=2$\,Mb/s, leading to a transaction throughput of about $2,000$\,tx/s.

Overall, our results show that the Blockclique architecture efficiently uses the underlying network, yielding high consensus bitrates relative to the bandwidth of nodes even in large networks.

\section{Discussion}

The Blockclique architecture combines three main ideas that together make scaling possible: the data structure is a multithreaded block DAG where each block references one previous block of each thread, transaction sharding separates transactions into multiple threads based on their input address so that blocks created independently in different threads have compatible transactions by construction, and the consensus rule extends Nakamoto consensus to leverage the parallel creation of blocks.

Blockclique is a simple extension of Nakamoto blockchains, where each node receives and verifies all blocks and transactions of all threads. 
As transactions are distributed into threads but all nodes process them, there is no need for cross-shard communications other than cross-thread parent links in the graph. 
Transaction sharding guarantees that no double-spend can happen even when blocks are created in parallel.

Parallel threads with minimal inter-thread synchronization requirements, and lack of multiple inclusions of the same transactions result in efficient, smooth and predictable network usage, bringing the consensus bitrate close to the capacity of the network.

Blockclique is secure against attacks aiming at tampering with transaction history, forking the network or denying service within a broad range of parameters, assuming that the proportion of resources owned by the attacker is below $45\%$.

Our network simulations show that a highly multithreaded block graph efficiently uses available network bandwidth and reaches a transaction throughput of more than $10,000$~tx/s with a stable consensus in a large decentralized network of thousands of nodes.
Based on network parameter estimates in Bitcoin and Ethereum, we assumed a median bandwidth of $32$\,Mb/s, but a different assumption would scale the maximum possible transaction throughput accordingly.
In principle, it is possible to dynamically adjust some aspects of the network such as the number of threads through a fork of the client code or an upgrade through a governance mechanism such as the amendment process in Tezos.
However, our results show that $T=32$ threads are suitable for a wide range of realistic bandwidth and latency values.

As Blockclique uses a ledger-based approach (instead of UTXOs), it is possible to store extra data and programs for each address, and design specific types of transactions that act on them, in order to implement a smart contract system. The sharding process however, requires that each smart contract lives in a specific thread, or uses sharding logic by itself.


\bibliographystyle{abbrv}
\bibliography{bibliography}

\end{document}